\documentclass{scrartcl}

\usepackage[normalem]{ulem}

\usepackage{todonotes}
\usepackage[utf8]{inputenc}

\usepackage[ruled]{algorithm2e}

\usepackage{tikz}
\usepackage{mathtools}
\mathtoolsset{showonlyrefs=true}

\usepackage{hyperref}
\usepackage{subfig,caption}

\usepackage{amsmath, amssymb, amsthm}
\newtheorem{remark}{Remark}
\newtheorem{proposition}{Proposition}

\newcommand{\freq}{\omega}
\newcommand{\acc}{\nu}
\newcommand{\slope}{\acc}

\DeclareMathOperator{\e}{e}
\DeclareMathOperator{\diff}{d}
\DeclareMathOperator{\idty}{Id}

\newcommand{\C}{\mathbb{C}}
\newcommand{\R}{\mathbb{R}}

\begin{document}

\title{A bio-inspired geometric model for sound reconstruction}


\author{Ugo Boscain\thanks{CNRS, LJLL, Sorbonne Universit\'e, Universit\'e de Paris, Inria,  Paris, France. \texttt{ugo.boscain@upmc.fr}}
    \and Dario Prandi\thanks{Universit\'e Paris-Saclay, CNRS, CentraleSup\'elec,  Laboratoire des signaux et syst\`emes, 91190, Gif-sur-Yvette, France. \texttt{dario.prandi@centralesupelec.fr}}
    \and Ludovic Sacchelli\thanks{Univ. Lyon, Universit\'e Claude Bernard Lyon 1, CNRS, LAGEPP UMR 5007, 43 bd du 11 novembre 1918, F-69100 Villeurbanne, France. \texttt{ludovic.sacchelli@univ-lyon1.fr}}
    \and Giuseppina Turco\thanks{CNRS, Laboratoire de Linguistique Formelle, Université de Paris, France. \texttt{gturco@linguist.univ-paris-diderot.fr}}
    }











\date{\today}


\maketitle


\begin{abstract}
The reconstruction mechanisms built by the human auditory system during sound reconstruction are still a matter of debate. The purpose of this study is to propose a mathematical model of  sound reconstruction based on the functional architecture of the auditory cortex (A1). The model is inspired by the geometrical modelling of vision, which has undergone a great development in the last ten years. 
There are however fundamental dissimilarities, due to the different role played by time and the different group of symmetries.
The algorithm transforms the degraded sound in an ‘image’ in the time-frequency domain via a short-time Fourier transform. Such an image is then lifted in the Heisenberg group and is reconstructed via a Wilson-Cowan differo-integral equation. Preliminary numerical experiments are provided, showing the good reconstruction properties of the algorithm on synthetic sounds concentrated around two frequencies. 
\end{abstract}


\section{Introduction}\label{s:intro}

Listening to speech requires the capacity of the auditory system to map incoming sensory input to lexical representations. When the sound is intelligible, this mapping (“recognition”) process is successful. With reduced intelligibility (e.g., due to background noise), the listener has to face the task of recovering the loss of acoustic information. This task is very complex as it requires a higher cognitive load and the ability of repairing missing input. (See \cite{Mattys} for a review on noise in speech.) Yet, (normal hearing) humans are quite able to recover sounds in several effortful listening situations (see for instance \cite{Luce}), ranging from sounds degraded at the source (e.g., hypoarticulated and pathological speech), during transmission (e.g., reverberation) or corrupted by the presence of environmental noise.

So far, work on degraded speech has informed us a lot on the acoustic cues that help the listener to reconstruct missing information (e.g.,\cite{Fernandes2007, Parikh}); the several adverse conditions in which listeners may be able to reconstruct speech sounds (e.g., \cite{Assmann2004,Mattys}); and whether (and at which stage of the auditory process) higher-order knowledge (i.e., our information about words and sentences) helps the system to recover lower-level perceptual information (e.g., \cite{HANNEMANN2007}). 
%
However, most of these studies adopt a phenomenological and descriptive approach. More specifically, techniques from previous studies consist in adding synthetic noise to speech sound stimuli, performing spectral and temporal analyses on the stimuli with noise and the same ones without it so to identify acoustic differences, linking the results of these analyses with the outcome from perceptual experiments. In some of these behavioural experiments, for instance, listeners are asked to identify speech units (such as consonants or words) when listening the noisy stimuli. Their accuracy scores provide a measure to the listeners' speech recognition ability.

As it stands, a mathematical model informing us on how the human auditory system is able to reconstruct a degraded speech sound is still missing. The aim of this study is to build a neuro-geometric model for sound reconstruction, stemming from the description of the functional architecture of the auditory cortex.

\subsection{Modelling the auditory cortex}

Knowledge about the functional architecture of the auditory cortex is scarce, and there are difficulties in the application of Gestalt principles for auditory perception. For these reasons, the model we propose is strongly inspired by recent advances in the mathematical modeling of the functional architecture of the primary visual cortex and the processing of visual inputs \cite{Hoffman, Petitot, Citti2006, boscain2010anthropomorphic}, which recently yield very successful applications to image processing \cite{Franken09, Duits2010, Prandi2017, remizovJMIV}.
This idea is not new: neuroscientists take models of V1 as a starting point for understanding the auditory system (see, e.g., \cite{Nelken2011} for a comparison, and \cite{Hickok} for a related discussion in speech processing). Indeed, biological similarities between the structure of the primary visual cortex (V1) and the primary auditory cortex (A1) are well-known to exist. 

An often cited V1-A1 similarity is their “topographic” organization, a general principle determining how visual and auditory inputs are mapped to those neurons responsible for their processing \cite{Rauschecker}. 
Substantial evidence for V1-A1 relation is also provided by studies on animals and on humans with deprived hearing or visual functions showing cross-talk interactions between sensory regions \cite{Sharma, ZATORRE200113}. 
More relevant for our study is the existence of receptive fields of neurons in V1 and A1 that allow for a subdivision of neurons in “simple” and “complex” cells, which supports the idea of a “common canonical processing algorithm within cortical columns” \cite[p.~1]{Tian7892}.  
Together with the appearance in  A1 of singularities typical of V1 (e.g., pinwheels)  \cite{Sharma, Polger}, these findings speak in favor of the idea that V1 and A1 share similar mechanisms of sensory input reconstruction. In the next section we present the mathematical model for V1 that will be the basis for our sound reconstruction algorithm.

\subsection{Neuro-geometric model of V1}\label{angiulinalabella}
The neuro-geometric model of V1 finds its roots in the experimental results of Hubel and Wiesel \cite{hubel-wiesel}, which inspired Hoffman \cite{Hoffman} to model V1 as a \emph{contact space}\footnote{
A $3$ dimensional manifold $M$ becomes a {\em contact space} once it is
endowed with a  smooth map $M\ni q\mapsto {\cal D}(q)$ where  ${\cal D}(q)$
is a a plane in the tangent space $T_q M$ passing from $q$.                                   
 There is an additional requirement on this map. Locally one can always                       
write  ${\cal D}(q)=\operatorname{span}\{X_0(q),X_1(q)  \}$, where $X_0$ and $X_1$
are two smooth vector fields. Then at every point $q$ one should require                      
dim(span$_q\{X_0,X_1,[X_0,X_1]\})=3$. Here $[\cdot,\cdot]$ is the Lie                         
bracket of the vector fields. The main consequence of this condition is                       
that no surface can be tangent to ${\cal D}$ at all points.                                   
                                                                                              
By assigning to every ${\cal D}(q)$ an inner (Euclidean) product that is 
smooth  as function of $q$, we endow $M$ with a {\em sub-Riemannian structure}. The simplest way of defining locally such a structure on a $3$ dimensional manifold is to assign two vector fields $X_0$ and                       
$X_1$ postulating  on one side that ${\cal D}(q)=\operatorname{span}\{X_0(q),X_1(q)  \}$               
(assigning in this way the contact structure) and on the other side that                      
they have norm one and that they are mutually orthogonal (assigning in this                   
way the inner product).                                                                       
                                                                                              
The simplest example of sub-Riemannian structure on $\R^3$ is given by the                    
so called Heisenberg group for which the vector fields $X_0 = (1,\nu,0)^\top$ and
$X_1 = (0,0,1)^\top$  are orthonormal (here we  write coordinates in  $\R^3$ as 
$(\tau,\omega,\nu)$). Such a structure is called Heisenberg group since               
defining $X_2=(0,1,0)^\top$ one has the Lie brackets                       
 $[X_0,X_1]=X_2$,  $[X_0,X_2]=[X_1,X_2]=0$, that  are                                        
 the commutation relations appearing in quantum mechanics.  
}. This model has then been extended to the so-called sub-Riemannian model in \cite{Petitot1999, Citti2006, boscain2010anthropomorphic, Remizov2013}. On the basis of such a model, exceptionally efficient algorithms for image inpainting have been developed (e.g., \cite{remizovJMIV,Duits, Duits2010}). These algorithms have now several medical imaging applications (e.g., \cite{zhang}).

The main idea behind this model is that an image, seen as a function $f:\R^2\to\R_+$ representing the grey level, is lifted to a distribution on $\R^2\times P^1$, the bundle of directions of the plane\footnote{Note that in mathematics, the term “direction” corresponds to what neurophysiologists call “orientation” and viceversa. In this study, we use the mathematical terminology}.
Here, $P^1$ is the projective line, i.e., $P^1 = \R/\pi\mathbb{Z}$. 
More precisely, the lift is given by $Lf(x,y,\theta) = \delta_{Sf}(x,y,\theta)f(x,y)$ where $\delta_{S_f}$ is the Dirac mass supported on the set $S_f\subset \R^2\times P^1$ of points $(x,y,\theta)$ such that $\theta$ is the direction of the tangent line to $f$ at $(x,y)$.
Notice that, under suitable regularity assumptions on $f$, $S_f$ is a surface. 

When $f$ is corrupted (i.e. when $f$ is not defined in some region of the plane), the lift is corrupted as well and the reconstruction is obtained by applying a deeply anisotropic diffusion adapted to the problem.  Such diffusion mimics the flow of information along the horizontal and vertical connections of V1 and uses as an initial condition the surface $S_f$ and the values of the function $f$. Mathematicians call such a diffusion the {\em sub-Riemannian  diffusion} in  $\R^2\times P^1$, cf.\ \cite{Montgomery2002, Agrachev2019}.
One of the main features of this diffusion is that it is invariant by rototranslation of the plane, a feature that will not be possible to translate to the case of sounds, due to the special role of the time variable.

In what follows, we explain how similar ideas could be translated to the problem of sound reconstruction. 

\subsection{From V1 to sound reconstruction}
\begin{figure}
    \centering
    \scalebox{.5}{
        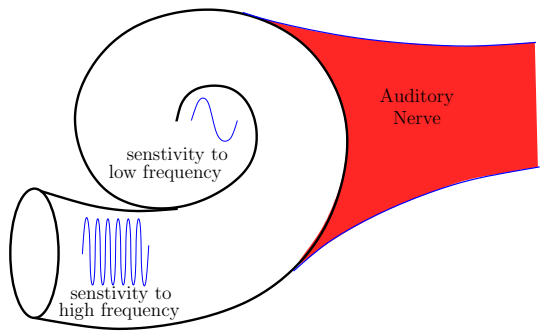
    }
\caption{Perceived pitch of a sound depends on the location in the cochlea that the sound wave stimulated. High-frequency sound waves, which correspond to high-pitched noises, stimulate the basal region of the cochlea. Low-frequency sound waves are targeted to the apical region of the cochlear structure and correspond with low-pitched sounds.}
  \label{fig:cochlea}
\end{figure}

The sensory input reaching A1 comes directly from the cochlea \cite{InnerEar}: a spiral-shaped, fluid-filled, cavity that composes the inner ear. Vibrations coming from the ossicles in the middle ear are transmitted to the cochlea, where they propagate and are picked up by sensors (so-called hair cells). These sensors are tonotopically organized along the spiral ganglion of the cochlea in a frequency-specific fashion, with cells close to the base of the ganglion being more sensitive to low-frequency sounds and cells near the apex more sensitive to high-frequency sounds, see Figure~\ref{fig:cochlea}. This early ‘spectrogram’ of the signal is then transmitted to higher-order layers of the auditory cortex. 

Mathematically speaking, this means that when we hear a sound (that we can think as represented by a function $s:[0,T]\to\R$) our primary auditory cortex A1 is fed by its time-frequency representation\footnote{Actually, its spectrogram $|S|:[0,T]\times \R\to [0,+\infty)$, see Remark~\ref{rmk:spectrogram}.} $S:[0,T]\times\R\to \C$. If, say, $s\in L^2(\R^2)$ the time-frequency representation $S$ is given by the short-time Fourier transform of $s$, defined as
\begin{equation}
  S(\tau, \omega) := \operatorname{STFT}(s)(\tau,\omega)= \int_{\R} s(t)W(\tau-t)e^{2\pi i t\omega}\,dt.
\end{equation}
Here, $W:\R\to [0,1]$ is a compactly supported (smooth) window, so that $S\in L^2(\R^2)$.
Since $S$ is complex-valued, it can be thought as the collection of two black-and-white images: $|S|$ and $\arg S$.
The function $S$ depends on two variables: the first one is time, that here we indicate with the letter $\tau$, and the second one is frequency, denoted by $\omega$. Roughly speaking, $|S(\tau,\omega)|$ represents the strength of the presence of the frequency $\omega$ at time $\tau$. In the following, we call $S$ the sound image (see Figure~\ref{fig:stft}). 

\begin{figure}
    \centering
  \includegraphics[width = .4\textwidth]{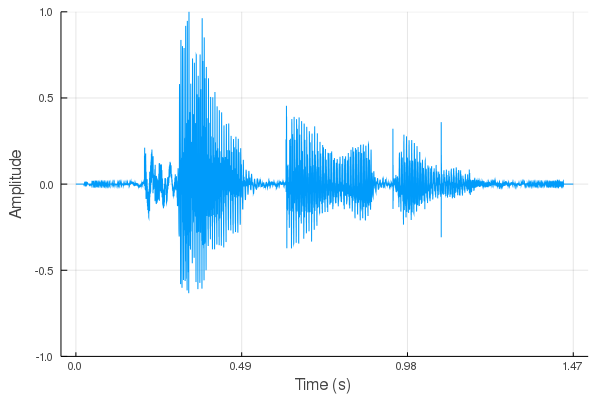}
  \includegraphics[width = .4\textwidth]{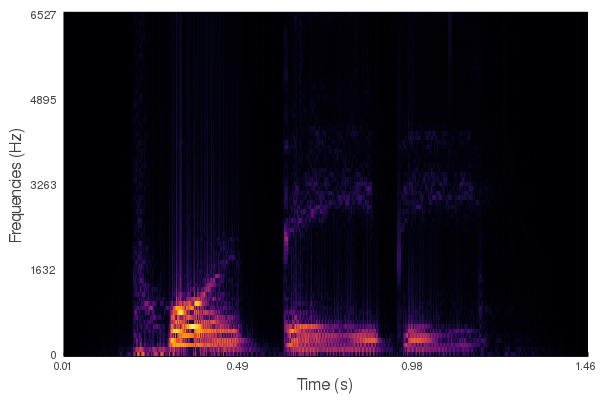}
  \caption{\emph{Left.} Sound signal. \emph{Right.} The corresponding short-time Fourier transform.}
  \label{fig:stft}
\end{figure}

A first attempt to model the process of sound reconstruction into A1 is to apply the algorithm for image reconstruction described in Section \ref{angiulinalabella}. 
In a sound image, however, time plays a special role. Indeed:
\begin{enumerate}
	\item While for images the reconstruction can be done by evolving the whole image simultaneously, the whole sound image does not reach the auditory cortex simultaneously, but sequentially. Hence, the reconstruction can be performed only in a sliding window.
	\item A rotated sound image corresponds to a completely different input sound and thus the invariance by rototranslations is lost.
\end{enumerate} 
\noindent As a consequence, different symmetries have to be taken into account (see Appendix~\ref{a:heisenberg}) and a different model for both the lift and the processing in the lifted space is required. 

In order to introduce this model, let us recall that, in V1, neural stimulation stems not only from the input but also from its variations. That is, mathematically speaking, the input image is considered as a real valued function on a 2-dimensional space, and the orientation sensitivity arises from the sensitivity to a first order derivative information on this function, i.e., the tangent directions to level lines. 
This additional variational information allows to lift the 2-dimensional image space to the aforementioned contact space, and to define the sub-Riemannian diffusion \cite{Agrachev2019,Bramanti2014}.

In our model of A1, we follow the same idea: we consider the variations of the input as additional variables. Input sound signals are time dependent real-valued functions subjected to a short time Fourier transform by the cochlea. As a result the A1 input is considered as a function of time and frequency. The first time derivative $\nu=d\omega/d\tau$ of this object, corresponding to the instantaneous chirpiness of the sound, allows to add a supplementary dimension to the domain of the input. As in the case of V1, this gives rise to a natural lift of the signal to an \emph{augmented space}, which in this case turns out to be $\mathbb{R}^3$ with the Heisenberg group structure. (This structure very often appears in signal processing, see for instance \cite{Karlheinz2001} and Appendix~\ref{a:heisenberg}.) 

As we already mentioned, the special role played by time in sound signals does not permit to model the flow of information as a pure hypoelliptic diffusion, as was done for static images in V1. We thus turn to a different kind of model: Wilson-Cowan equations \cite{Wilson1972}. Such a model, based on an integro-differential equation, has been successfully applied to describe the evolution of neural activations.  In particular, it allowed to theoretically predict complex perceptual phenomena in V1, such as the emergence of hallucinatory patterns \cite{Ermentrout, Bressloff01}, and has been used in various computational models of the auditory cortex \cite{Loebel2007, Rankin2015, Zulfiqar2020}. Recently, these equations have been coupled with the neuro-geometric model of V1 to great benefit. For instance, in \cite{SSVM2019, JNP2019, BertalmioCalatroniEtAl2019} they allowed to replicate orientation-dependent brightness illusory phenomena, which had proved to be difficult to implement for non-cortical-inspired models. See also \cite{Sarti2014}, for applications to the detection of perceptual units.

	On top of these positive results, Wilson-Cowan equations present many advantages from the point of view of A1 modelling: i) they can be applied independently of the underlying structure, which is only encoded in the kernel of the integral term; ii) they allow for a natural implementation of delay terms in the interactions; iii) they can be easily tuned via few parameters with a clear effect on the results.
On the basis of these positive results, we emulate this approach in the A1 context.
Namely, we will consider the lifted sound image $I(\tau,\freq,\slope)$ to yield an A1 activation $a(\tau,\freq,\slope)$ via the following Wilson-Cowan equations:
\begin{multline}\tag{WC}\label{eq:WC}
        \partial_t a(t,\freq,\acc) = {-\alpha a(t,\freq,\acc)}+{\beta I(t,\freq,\acc)}\\
        +{\gamma \int_{\mathbb R^2} k_\delta(\freq,\acc\|\freq',\acc')\sigma(a(t-\delta, \freq',\acc'))\,d\freq'\,d\acc'}.
\end{multline}
Here, $(t,\omega,\nu)$ are coordinates on the augmented space corresponding to time, frequency, and chirpiness, respectively; $\alpha, \beta, \gamma >0$ are parameters; $\sigma:\C\to \C$ is a non-linear sigmoid; $k_\delta(\freq,\acc\|\freq',\acc')$ is a weight modelling the interaction between $(\freq,\acc)$ and $(\freq',\acc')$ after a delay of $\delta>0$.
The presence of this delay term models the fact that the time-scale of the input signal and of the neuronal activation are comparable.

The proposed algorithm to process a sound signal $s:[0,T]\to \R$, is the following:
\begin{enumerate}
  \item[A.] \textbf{Preprocessing:} 
  \begin{enumerate}
    \item Compute the time-frequency representation $S:[0,T]\times \R\to \C$ of $s$, via standard short time Fourier transform (STFT);
    \item Lift this representation to the Heisenberg group, which encodes redundant information about chirpiness, obtaining $I:[0,T]\times\R\times\R\to \C$ (see Section~\ref{sec:lift} for details);
  \end{enumerate}
  \item[B.] \textbf{Processing:} Process the lifted representation $I$ via a Wilson-Cowan equations adapted to the Heisenberg structure, obtaining $a:[0,T]\times\R\times\R\to \C$.
  \item[C.] \textbf{Postprocessing:} Project $a$ to the processed time-frequency representation $\hat S:[0,T]\times \mathbb R\to \mathbb C$ and then apply an inverse STFT to obtain the resulting sound signal $\hat s:[0,T]\to \R$.
\end{enumerate}

\begin{remark}
  All the above operations can be performed in real-time, as they only require the knowledge of the sound on a short window $[t-\delta,t+\delta]$.
\end{remark}

\begin{remark}\label{rmk:spectrogram}
    Notice that in the presented algorithm we are assuming neural activations to be complex-valued functions,  due to the use of the STFT. This is inconsistent with neural modelling, as it is known that the cochlea sends to A1 only the spectrogram of the STFT (that is $|S|$), see \cite{sethares2005tuning}.
    If striving for a biologically plausible description, one can easily modify the above algorithm in this direction (i.e., by computing the lifted representation $I$ starting from $|S|$ instead than $S$). However, during the post-processing phase, in order to invert the STFT and obtain an audible signal, one then needs to reconstruct the missing phase information via heuristic algorithms. See, for instance \cite{Fienup1982}.
\end{remark}

\subsection{Structure of the paper}

In Section~\ref{sec:reconstruction}, we present the reconstruction model. We first present the lift procedure of a sound signal to a function on the augmented space, and then introduce the Wilson-Cowan equations modelling the cortical stimulus. In Section~\ref{sec:numerics}, we describe the numerical implementation of the algorithm, together with some of its  crucial properties. 
This implementation is then tested in Section~\ref{s:experiments}, were we show the results of the algorithm on some simple synthetic signals. Such numerical examples can be listened at \url{www.github.com/dprn/WCA1}, and should be considered as a very preliminary step toward the construction of an efficient cortical-inspired algorithm for sound reconstruction. 
Finally, in Appendix~\ref{a:heisenberg}, we show how the proposed algorithm preserves the natural symmetries of sound signals.

\section{The reconstruction model}
\label{sec:reconstruction}

As discussed in the introduction, the cochlea decomposes the input sound $s:[0,T]\to \R$ in its time-frequency representation $S:[0,T]\times \R \to \C$, obtained via a short-time Fourier transform (STFT). This corresponds to interpret the ``instantaneous sound'' at time $\tau\in[0,T]$, instead of as a sound level $s(\tau)\in \R$, as a function $\freq\mapsto S(\tau,\freq)$ which encodes the instantaneous presence of each given frequency, with phase information.

\subsection{The lift to the augmented space}
\label{sec:lift}

In this section, we present an extension of the time-frequency representation of a sound, which is at the core of the proposed algorithm. Roughly speaking, the instantaneous sound will be represented as a function $(\freq, \slope)\mapsto I(\tau,\freq,\slope)$, encoding the presence of both the frequency and the chirpiness
$\slope = {d\freq}/{d\tau}$.

Assume for the moment that the sound has a single time-varying frequency, e.g.,
\begin{equation}\label{eq:sin-omega}
    s(\tau)=A \sin(\freq(\tau)\tau), \qquad A\in\R.
\end{equation}
If the frequency is varying slowly enough and the window of the STFT is large enough, its sound image (up to the choice of normalisations constants in the Fourier transform) coincides roughly with
\begin{equation}
    S(\tau,\freq) \sim \frac{A}{2i} \bigg( \delta_0(\freq-\freq(\tau))-\delta_0(\freq+\freq(\tau)) \bigg),
\end{equation}
where $\delta_0$ is the Dirac delta distribution centered at $0$. That is, $S$ is concentrated on the two curves $\tau\mapsto (\tau,\freq(\tau))$ and $\tau\mapsto (\tau,-\freq(\tau))$, see Figure~\ref{fig:single-freq}. Let us focus only on the first curve.

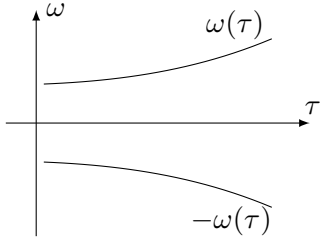
\begin{figure}
    \centering

\begin{tikzpicture}[x=1.cm,y=1.cm]
\draw [domain=1:4] plot (\x,\x*\x/100+\x*\x*\x/140+.5);
\draw [domain=1:4] plot (\x,-\x*\x/100-\x*\x*\x/140-.5);
\draw[-latex] (0.5,0) -- (4.5,0);
\draw[-latex] (.9,-1.5) -- (.9,1.5);
\draw (3,1.3) node[right] {$\omega(\tau)$};
\draw (2.8,-1.3) node[right] {$-\omega(\tau)$};
\draw (4.3,.2) node[right] {$\tau$};
\draw (.9,1.5) node[right] {$\omega$};
\end{tikzpicture}

\caption{Short-time Fourier transform of the signal in \eqref{eq:sin-omega}, for a positive and increasing $\omega(\cdot)$.}
\label{fig:single-freq}
\end{figure}

Because of the sensitivity to variations of the input, as discussed in Section~\ref{s:intro}, the curve $\freq(\tau)$ is lifted in a bigger space by adding a new variable $\nu=d\freq/d\tau$.  In mathematical terms the 3-dimensional space $(\tau,\omega,\nu)$ is called the \emph{augmented space}. It will be the basis for the geometric model of A1 that we are going to present.
 
Up to now the curve $\omega(\tau)$ was parameterized by one of the coordinates of the contact space (the variable $\tau$), but it will be more convenient to consider it as a parametric curve in the space $(\tau,\omega)$. More precisely, the original curve $\omega(\tau)$ is represented in the space $(\tau,\omega)$ as $t\mapsto(t,\omega(t))$ (thus imposing $\tau=t$). Similarly, the lifted curve is parameterized as $t\mapsto(t,\omega(t),\nu(t))$. To every regular enough curve $t\mapsto(t,\omega(t))$, one can associate a lift $t\mapsto(t,\omega(t),\nu(t))$ in the contact space simply by computing $\nu(t)=d\omega/dt$. Vice-versa, a regular enough curve in the contact space 
$t\mapsto(\tau(t),\omega(t),\nu(t))$ is a lift of planar curve $t\mapsto (t,\omega(t))$ if $\tau(t)=t$  and if $\nu(t)=d\omega/dt$. Now, defining $u(t)=d\nu/dt$  we can say that a curve in the contact space $t\mapsto(\tau(t),\omega(t),\nu(t))$ is a lift of a planar curve if there exists a function $u(t)$   such that:
\begin{equation}\label{eq:heis}
  \frac{d}{dt}
 \begin{pmatrix}
    \tau \\ \freq \\ \acc 
 \end{pmatrix}
  =
	\begin{pmatrix}
	    1 \\ \acc \\ 0
\end{pmatrix}
  +u(t)
	\begin{pmatrix}
	0 \\ 0 \\ 1
	\end{pmatrix}.
\end{equation}
Letting $q=(\tau,\omega,\nu)$, equation \eqref{eq:heis} can be equivalently written as the control system
\begin{equation}\label{eq:heis-vf}
\frac{d}{dt}q(t)=X_0(q(t))+u(t) X_1(q(t)),
\end{equation}
where the $X_0$ and $X_1$ are the two vector fields in $\R^3$
\begin{equation}
X_0=
  \begin{pmatrix}
    1 \\ \acc \\ 0
  \end{pmatrix}, ~~~X_1=
  \begin{pmatrix}
    0 \\ 0 \\ 1
  \end{pmatrix}.
\end{equation}
Notice that the two vector fields appearing in this formula generate the Heisenberg group. However, we are not dealing here with a sub-Riemannian structure, since the space $\{X_0+uX_1\mid u\in\R\}$ is a line and not a  plane. (One would get a plane by considering two controls, namely $\{u_0X_0+u_1X_1\mid(u_0,u_1)\in\R^2\}$.)     

Following \cite{boscain2010anthropomorphic}, when $s$ is a general sound signal, we lift each level line of $|S|$. By the implicit function theorem, this yields the following subset of the contact space:
\begin{equation}\label{eq:Sigma}
    \Sigma = 
    \left\{ (\tau,\freq,\slope)\in\R^3\mid \slope \partial_\freq |S|(\tau,\freq)+\partial_\tau|S|(\tau,\freq) =0\right\}.
\end{equation}
If $|S|\in C^2$ and $\operatorname{Hess}|S|$ is non-degenerate, the set $\Sigma$ is indeed a surface.
Finally, the external input from the cochlea is  given by
\begin{equation}\label{eq:lift-distr}
    I(\tau,\freq,\slope) = S(\tau,\freq)\delta_{\Sigma}(\tau,\freq,\slope).
\end{equation}
Here, $\delta_\Sigma$ denotes the Dirac delta distribution concentrated at $\Sigma$. The presence of this distributional term is necessary for a well-defined solution to the evolution equation \eqref{eq:WC}. Such en equation is introduced in the next section.

\subsection{Cortical activations in A1}\label{s:WC}

On the basis of what described in the previous section and the well-known tonotopical organization of A1 (cf.\ Section~\ref{s:intro}), we propose to consider A1 to be the space of $(\freq,\acc)\in \R^2$. When hearing a sound $s(\cdot)$, the external input fed to A1 at time $t>0$ is then given as the slice at $\tau=t$ of the lift $I$ of $s$ to the contact space.
That is, hearing an ``instantaneous sound level'' $s(t)$ reflects in the external input $I(t,\freq,\acc)$ to the ``neuron'' $(\freq,\acc)$ in A1 as follows: The ``neuron'' receives an external charge $S(t,\omega)$ if $(t,\omega,\nu)\in \Sigma$, and no charge otherwise, where $\Sigma$ is defined in \eqref{eq:Sigma}.

We model the neuronal activation induced by the external stimulus $I$ by adapting to this setting the well-known Wilson-Cowan equations. These equations are widely used and proved to be very effective in the study of V1 \cite{Wilson1972, Bressloff01}.
According to this framework, the resulting activation $a:[0,T]\times \R\times \R\to \C$ is the solution of the following equation with delay $\delta>0$:
\begin{equation}\label{eq:WC1}
        \partial_t a(t,\freq,\acc) = {-\alpha a(t,\freq,\acc)}+{\beta I(t,\freq,\acc)}
        +{\gamma \int_{\mathbb R^2} k_\delta(\freq,\acc\|\freq',\acc')\sigma(a(t-\delta, \freq',\acc'))\,d\freq'\,d\acc'},
\end{equation}
with initial condition $a(t,\cdot,\cdot)\equiv 0$ for $t\le 0$.
Here, $\alpha,\beta, \gamma>0$ are parameters, $k_\delta$ is an interaction kernel, and $\sigma:\C\to \C$ is a (non-linear) saturation function, or sigmoid. In the following, we let $\sigma(\rho e^{i\theta})=\tilde\sigma(\rho) e^{i\theta}$ where $\tilde \sigma(x) = \min\{1,\max\{0,\kappa x\}\}$, $x\in\R$, for some fixed $\kappa>0$. The fact that the non-linearity $\sigma$ does not act on the phase is one of the key ingredients in proving that this processing preserves the natural symmetries of sound signals, see Proposition~\ref{prop:symmetry} in Appendix~\ref{a:heisenberg}.

When $\gamma = 0$, equation \eqref{eq:WC1} becomes the standard low-pass filter $\partial_t a = -\alpha a + I$, whose solution is the convolution of the input signal $I$ with the function
\begin{equation}
    \varphi(t) = 
    \begin{cases}
        e^{-t\alpha} & \quad\text{if } t>0,\\
        0 &\quad\text{otherwise}.
    \end{cases}
\end{equation}
Setting $\gamma\neq 0$ adds a non-linear delayed interaction term on top of this exponential smoothing, encoding the inhibitory and excitatory interconnections between neurons. Next section is devoted to the choice of the integral kernel $k_\delta$.

\begin{remark}
    In \eqref{eq:WC1} we chose to consider a simple form for the interaction term.
    A more precise choice would indeed need to take into account the whole history of the process, for example by considering
    \begin{equation}
        \int_{\tau}^{+\infty} e^{-\varrho(s-\tau)} \int_{\R^2}k_{s}(\freq,\acc\|\freq',\acc') \sigma(a(t-s,\freq',\acc'))\,d\freq'\,d\acc'\,ds,
        \qquad\varrho>0.
    \end{equation}
\end{remark}

\subsection{The neuronal interaction kernel}\label{s:tonotopical}

Considering A1 as a slice of the augmented space allows to deduce a natural structure for neuron connections as follows. Going back to a sound composed by a single time-varying frequency $t\mapsto\omega(t)$, we have that its lift is concentrated on the curve $t\mapsto (\omega(t),\nu(t))$, such that
\begin{equation}\label{eq:contr-sys}
    \frac{d}{dt}\begin{pmatrix}
         \omega  \\
         \nu
    \end{pmatrix} 
    =
    Y_0(\omega,\nu) + u(t) Y_1(\omega,\nu),
\end{equation}
where $Y_0(\omega,\nu) = (\nu,0)^\top$, $Y_1(\omega,\nu) = (0,1)^\top$, and $u:[0,T]\to \R$.

    As in the case of V1 \cite{Remizov2013}, we model neuronal connections via these dynamics.
In practice, this amounts to assume that the excitation starting at a neuron $X_0=(\freq',\acc')$ evolves as the stochastic process $\{A_t\}_{t\ge 0}$ naturally associated with \eqref{eq:contr-sys}. This is given by the following stochastic differential equation,
    \begin{equation}\label{eq:stochastic}
        dA_t = Y_0(A_t) dt + Y_1(A_t) dW_T, \qquad A_0 = (\freq',\acc'),
    \end{equation}
    where $\{W_t\}_{t\ge 0}$ is a Wiener process. 
    The generator of $\{A_t\}_{t\ge 0}$ is the second order differential operator 
    \begin{equation}
        \mathcal L = Y_0 +(Y_1)^2 =\acc \partial_\freq + b\partial_\acc^2.
    \end{equation}
    In this formula, the vector fields $Y_0$ and $Y_1$ are interpreted as first-order differential operators.
    Moreover, we added a scaling parameter $b>0$, modelling the relative strength of the two terms. 

    It is natural to model the influence $k_\delta(\freq,\acc\|\freq',\acc')$ of neuron $(\freq',\acc' )$ on neuron $(\freq,\acc)$ at time $\delta>0$ as the transition density of the process $\{A_t\}_{t\ge0}$. It is well-known that such transition density is obtained by computing the integral kernel at time $\delta$ of the Fokker-Planck equation corresponding to \eqref{eq:stochastic}, that reads
    \begin{equation}
        \label{eq:kolmo}
	\partial_t I = \mathcal L^* I , 
	\quad\text{where}\quad
	\mathcal L^* =- Y_0+(Y_1)^2 = -\acc \partial_\freq + b\partial_\acc^2.
    \end{equation}

    The existence of an integral kernel for \eqref{eq:kolmo} is a consequence of the hypoellipticity\footnote{That is, if $f$ is a distribution defined on an open set $\Omega$ and such that $(\partial_t-\mathcal L^*)f\in C^\infty(\Omega)$, then $f\in C^\infty(\Omega)$.} of $(\partial_t - \mathcal L^*)$. The explicit expression of $k_\delta$ is well-known and we recall it in the following result, proved in Appendix~\ref{a:kernel}.

\begin{proposition}\label{p:expl-kernel}
	The integral kernel of equation \eqref{eq:kolmo} is
	\begin{equation}\label{eq:kernel}
	k_\delta(\freq,\nu\|\freq',\nu')=\frac{\sqrt{3}}{2 \pi b  \delta^2} \exp \left(-\frac{g_\delta(\freq,\acc\|\freq',\acc')}{
   b \delta^3}\right),
\end{equation}
where,
\begin{equation}
   g_\delta(\freq,\acc\|\freq',\acc') = 3
   (\freq-\freq')^2- 3
   \delta (\freq-\freq')
   (\nu+\nu')+ \delta^2
   \left(\nu^2+\nu
   \nu'+\nu'^2\right).
\end{equation}
\end{proposition}

\section{Numerical implementation}\label{sec:numerics}

For the numerical experiments, we chose to implement the proposed algorithm in Julia \cite{bezanson2017julia}. As already presented, this process consists in a preprocessing phase, in which we build an input function $I$ on the 3D contact space, a main part, where $I$ is used as the input of the Wilson-Cowan  equation \eqref{eq:WC}, and a post-processing phase, where the reconstructed sound is recovered from the result of the first part. 

In the following, we present these phases separately.

\subsection{Pre-processing} The input sound $s$ is lifted to a time-frequency representation $S$ via a classical implementation of STFT, i.e., by performing FFTs of a windowed discretised input.
In the proposed implementation we chose to use a standard Hann  window (see, e.g., \cite{Press2007}):
\begin{equation}
    W(x) =
	\begin{cases}
		\displaystyle\frac{1+\cos(2\pi x/L)}{2}&\text{ if } |x|< L/2,\\
		0 &\text{ otherwise.}
	\end{cases}
\end{equation}
The resulting time-frequency signal is then lifted to the contact space through an approximate computation of the gradient $\nabla |S|$ and the following discretisation of \eqref{eq:lift-distr}:
\begin{equation}\label{eq:disc-chirp}
    I(\tau, \freq,\acc) = 
    \begin{cases}
        S(\tau,\freq)& \text{if }\nu\partial_\freq|S|(\tau,\freq) =-\partial_\tau|S|(\tau,\freq),\\
        0 &\text{otherwise}.
    \end{cases}
\end{equation}

{
\paragraph*{Discretisation issues}

While the discretisation of the time and frequency domains is a well understood problem, dealing with the additional chirpiness variable requires some care. Indeed, even if we assume that the significant frequencies of the input sound $s$ belong to a bounded interval $\Lambda\subset \R$, in general the set $\{ \nu\in \R \mid I(\tau,\omega,\nu)\neq 0\}$ is unbounded. Indeed, one can check that as $(\tau,\freq)$ moves to a point where the countour lines of $|S|$ become vertical, the set of chirpinesses $\nu$'s such that $\nu\partial_\freq|S|(\tau,\freq) =-\partial_\tau|S|(\tau,\freq)$ will converge to $\pm \infty$.

In the numerical implementation we chose to restrict the admissible chirpinesses to a bounded interval $N\subset \R$. This set is chosen in a case by case fashion in order to contain the relevant slopes for the examples under consideration. Work is ongoing to automate this procedure.

}

\subsection{Processing}
Equation \eqref{eq:WC} can be solved via a standard forward Euler method. Hence, the delicate part of the numerical implementation is the computation of the interaction term. 

As is clear from the explicit expression given in Proposition~\ref{p:expl-kernel}, $k_\delta$ is not a convolution kernel. That is, $k_\delta(\freq,\acc\|\freq',\acc')$ cannot be expressed as a function of $(\freq-\freq',\acc-\acc')$. 
As a consequence, a priori we need to explicitly compute all values $k_\delta(\freq,\acc\|\freq',\acc')$ for  $(\freq,\acc)$ and $(\freq',\acc')$ in the considered domain. As  is customary, in order to reduce computation times, we fix a threshold $\varepsilon>0$ and for any given $(\freq,\acc)$ we compute only values for $(\freq',\acc')$ in the compact set
\begin{equation}
	\mathrm K_\delta^\varepsilon(\freq,\acc) = \{ (\freq',\acc')\mid  k_\delta(\freq,\acc\|\freq',\acc')\ge \varepsilon\}.
\end{equation}
The structure of $\mathrm K_\delta^\varepsilon(\freq,\acc)$ is given in the following, whose proof we defer to Appendix~\ref{a:kernel}.

\begin{proposition}\label{p:threshold}
For any $\varepsilon>0$ and $(\freq,\acc)\in \R^2$, we have that $\mathrm K_\delta^\varepsilon(\freq,\acc)$ is the set of those $(\freq',\acc')\in \R^2$ that satisfy
	\begin{gather}
		|\acc-\acc'|^2\le C_\varepsilon:=-4b\delta\log\left( \frac{2\pi b\tau^2}{\sqrt 3}\varepsilon\right),\\
		\left|\freq' -\freq+\frac{\delta(\acc + \acc')}{2}\right| 
		\le \frac{\delta}{2\sqrt{3}}\sqrt{C_\varepsilon-|\acc-\acc'|^2}.
	\end{gather}
\end{proposition}

\begin{remark}
	It holds $C_\varepsilon\ge0$ if and only if
	\begin{equation}
		\varepsilon\le \frac{\sqrt{3}}{2\pi b\delta^2}.
	\end{equation}
	Indeed, for any $(\freq,\acc)\in \R^2$, the r.h.s.\ above corresponds to $\max k_\delta(\freq,\acc\|\cdot,\cdot)$, and thus $\mathrm K^\varepsilon(\freq,\acc) = \varnothing$ for larger values of $\varepsilon$.
\end{remark}

The above allows to numerically implement $k_\delta$ as a family of sparse arrays. That is, let $G\subset \Lambda\times N$ be the chosen discretisation of the significant set of frequencies and chirpinesses. Then, to $\xi = (\freq, \acc)\in G$ we associate the array $M_{\xi}:G \to \R$ defined by
\begin{equation}
	M_{\xi}(\xi') =
	\begin{cases}
		k_\delta(\xi\|\xi') & \: \text{if }\xi' \in \mathrm K^\varepsilon(\xi),\\
		0& \: \text{otherwise.}
	\end{cases}
\end{equation}
Therefore, up to choosing the tolerance $\varepsilon\ll 1$ sufficiently small, the interaction term in \eqref{eq:WC}, evaluated at $\xi=(\freq,\acc)\in G$, can be efficiently estimated by
\begin{equation}
	 \int_{\mathbb R^2} w(\xi\|\xi')\sigma(a(t-\delta, \xi'))\,d\xi' \approx
	 \sum_{\xi'\in K^\varepsilon(\xi)} M_\xi(\xi')a(t-\delta,\xi').
\end{equation}

\subsection{Post-processing}
Both operations in the pre-processing phase are inversible: the STFT by inverse STFT, and the lift by integration along the $\acc$ variable (that is, summation of the discretized solution). The final output signal is thus obtained by applying the inverse of the pre-processing (integration then inverse STFT) to the solution $a$ of \eqref{eq:WC}. That is, the resulting signal is given by
\begin{equation}
  \hat s(t) = \operatorname{STFT}^{-1}
\left(
\int_{-\infty}^{+\infty} a(t,\freq,\acc)\diff \acc
\right).
\end{equation}

The following guarantees that $\hat s$ is real-valued and thus correctly represents a sound signal. From the numerical point of view, this implies that we can focus on solutions of \eqref{eq:WC} in the half-space $\{\freq\geq 0\}$, which can then be extended to the whole space by mirror symmetry.

\begin{proposition}
It holds that $\hat s(t)\in \R$ for all $t>0$.
\end{proposition}

\begin{proof}
Let us denote
\begin{equation}
\hat S(t,\freq)=\int_{-\infty}^{+\infty} a(t,\freq,\acc)\diff \acc,
\end{equation}
so that $\hat s = \operatorname{STFT}^{-1}(\hat S)$. Moreover, for any function $f(t,\freq,\acc)$, we let $f^\star(t,\freq,\acc):=\bar f(t,-\freq,-\acc)$.

To prove the statement it is enough to show that
\begin{equation}\label{eq:star}
    a(t,\cdot,\cdot)\equiv a^\star(t,\cdot,\cdot) \qquad \forall t\ge0.
\end{equation}
This is trivially satisfied for $t\le 0$, since in this case $a(t,\cdot,\cdot)\equiv 0$. 

We now claim that if \eqref{eq:star} holds on $[0,T]$ it holds on $[0,T+\delta]$, which will prove it for all $t\ge 0$.
%
By definition of $I$ and the fact that $S(t,-\freq)=\overline{S(t,\freq)}$, we immediately have $I \equiv I^\star$.
On the other hand, the explicit expression of $k_\delta$ in \eqref{eq:kernel} yields that 
\begin{equation}
  k_\delta(-\freq,-\acc\|\freq',\acc') = k_\tau(\freq,\acc\|-\freq',-\acc').
\end{equation}
Then, for all $t\le T+\delta$, we have
\begin{equation*}
  \begin{split}
      \int_{\mathbb R^2} &k_{\delta}(-\freq,-\acc\|\freq',\acc')\sigma(a(t-\tau, \freq',\acc'))\,d\freq'\,d\acc'\\
  &=
\int_{\mathbb R^2}
k_{\delta}(\freq,\acc\|\freq'',\acc'')\sigma(a^\star(t-\tau, \freq'',\acc''))\,d\freq''\,d\acc''\\
&=
\int_{\mathbb R^2}
k_{\delta}(\freq,\acc\|\freq'',\acc'')\sigma(a(t-\tau, \freq'',\acc''))\,d\freq''\,d\acc''.
  \end{split}
\end{equation*}
%
A simple argument, e.g., using the variation of constants method, shows that these two facts imply the claim, and thus the statement.
\end{proof}

\section{Experiments}\label{s:experiments}

In Figure~\ref{fig:bars-wave} we present a series of experiments on simple synthetic sounds in order to exhibit some key features of our algorithm. These experiments can be reproduced via the code available at \url{https://www.github.com/dprn/WCA1}.

The first example, Figure~\ref{F:linear}, is a simple linear chirp, such that the dominating frequency  depends linearly on time  (i.e., corresponding to $\omega(t)=\mu t$ for some $\mu\in \mathbb{R}$). One observes that the processed sound presents the same feature but for a longer duration. The parameters in the experiment have been chosen to emphasize the effect of the modelling equation: the reconstruction should not present a tail that is as pronounced, however this allows to highlight the diffusive effect along the lifted slope.

The second example, Figure~\ref{F:interrupted}, corresponds to the same linear chirp as Figure~\ref{F:linear}, that has been interrupted in its middle section, creating two disjoint linear chirps. Thanks to the transport effect of the algorithm, the gap between the two chirps is bridged in the processed signal.
For this illustration, the interruption lasts about twice as long as the delay. 

The third example, Figure~\ref{F:2bars}, consists of the sum of two linear chirps with different slopes. The slopes have been picked to suggest that linear continuations of the chirps should intersect. This is indeed what happens in the processed signal.
However, notice that the resulting crossing happens almost as a sum of the two chirps processed independently, with close to no interaction at the crossing. This is purely an effect of the lift procedure. The increasing chirp is (predominantly) lifted to a stratum corresponding to a positive slope, while the decreasing chirp is lifted to a negative slope stratum. De facto, their evolution under the Wilson--Cowan equation is decoupled in the 3D augmented space.

The fourth and last example, Figure~\ref{F:sin}, corresponds to a nonlinear chirp, roughly corresponding to choosing $\omega(\tau) = \sin (m \tau)$ in \eqref{eq:sin-omega}. 
The construction of the model favors linearity in the evolution of perceived frequencies. We can observe how the more linear elements of the input result in more diffusion.

\begin{figure}[b]
\centering

\subfloat[
				Linear chirp. (Parameters:  $\alpha=55$, $\beta=1$, $\gamma=55$, $b=0.05$.)\label{F:linear}
				]{
					  \includegraphics[width = .475\textwidth]{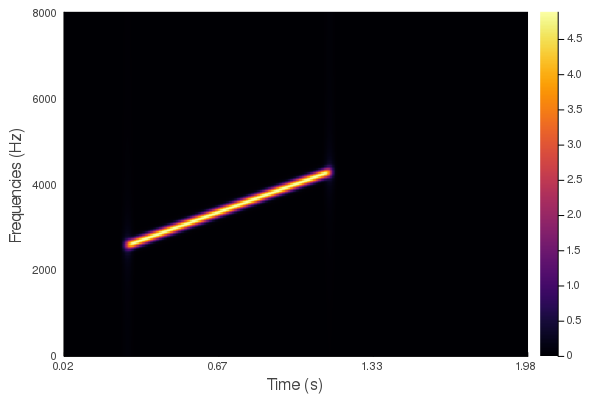}
					  \hspace{.02\linewidth}
					  \includegraphics[width = .475\textwidth]{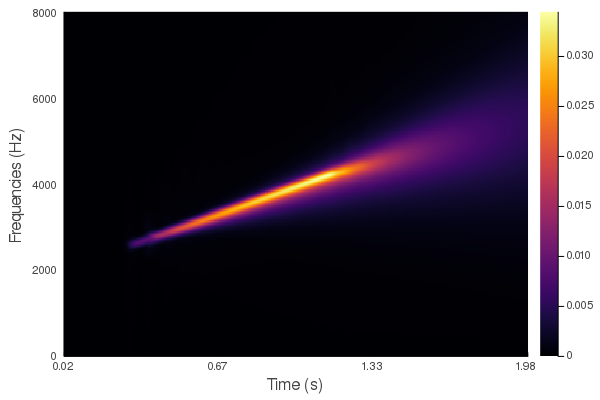}
 				}
 				
\subfloat[
				Interrupted chirp. (Parameters:  $\alpha=55$, $\beta=1$, $\gamma=55$, $b=0.05$.) \label{F:interrupted}
				]{
					  \includegraphics[width = .475\textwidth]{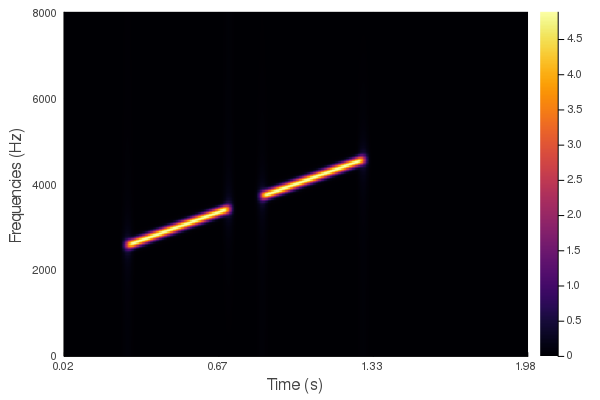}
					  \hspace{.02\linewidth}
					  \includegraphics[width = .475\textwidth]{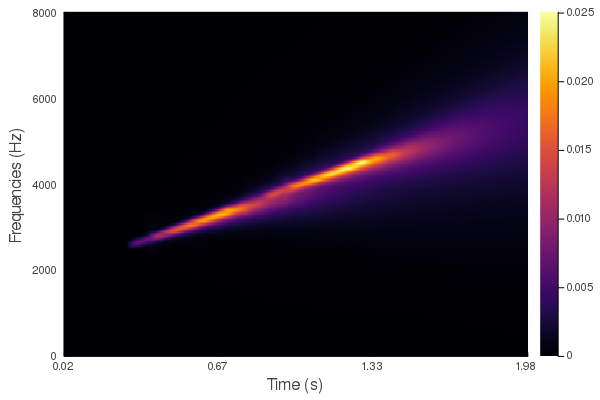}
 				} 				

\caption{
Experiments on synthetic sounds with varying frequencies. \emph{Left:} The STFT of the original sound. \emph{Right:} The STFT of the processed sound with delay $\delta=0.0625s$. Parameters vary for each experiment depending on the desired effect we wish to highlight. Each time, only the positive frequencies are shown: negative frequencies are recovered via the Hermitian symmetry of the Fourier transform on real signals.
%
}
\label{fig:bars-wave}

\end{figure} 

\begin{figure}[t]
 \ContinuedFloat		

\subfloat[
				Intersecting chirps. (Parameters:  $\alpha=53$, $\beta=1$, $\gamma=55$, $b=0.01$.) \label{F:2bars}
				]{
					  \includegraphics[width = .475\textwidth]{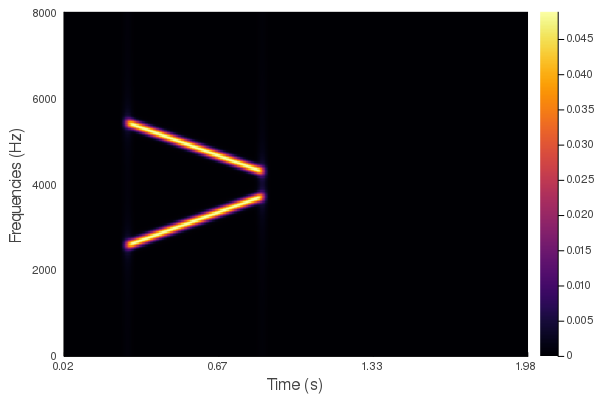}
					  \hspace{.02\linewidth}
					  \includegraphics[width = .475\textwidth]{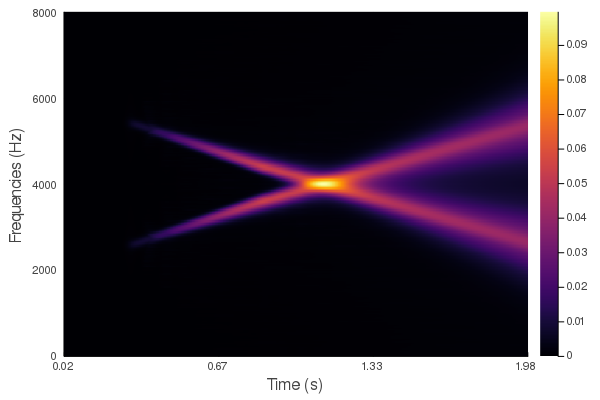}
 				}

\subfloat[
				Nonlinear chirp. (Parameters:  $\alpha=53$, $\beta=1$, $\gamma=55$, $b=0.2$.) \label{F:sin}
				]{
					  \includegraphics[width = .475\textwidth]{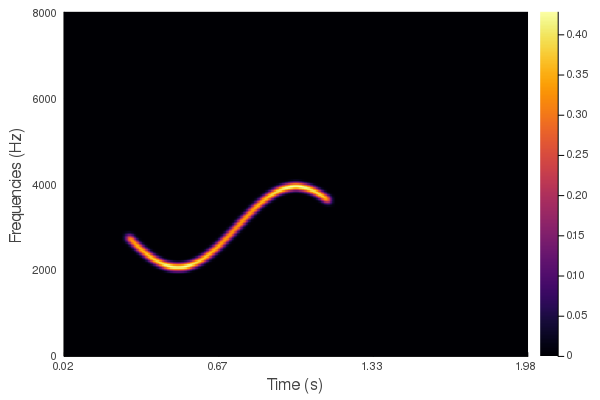}
					  \hspace{.02\linewidth}
					  \includegraphics[width = .475\textwidth]{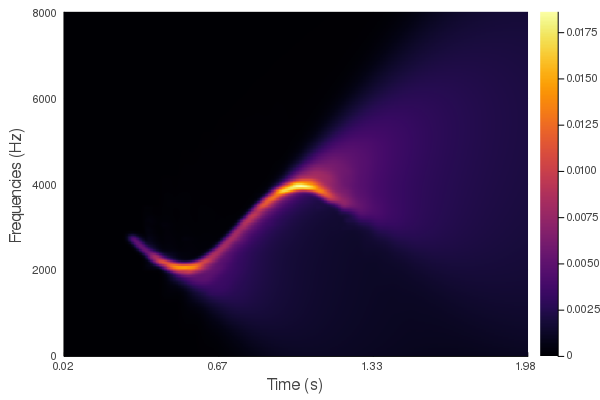}
 				} 				

\caption{Experiments on synthetic sounds with varying frequencies. (Continued.) 
}

\end{figure}

\section{Conclusion}

In this work we presented a sound reconstruction framework inspired by the analogies between visual and auditory cortices.
Building upon the successful cortical inspired image reconstruction algorithms, the proposed framework lifts time-frequencies representations of signals to the 3D contact space, by adding instantaneous chirpiness information. These redundant representations are then processed via adapted diffeo-integral Wilson-Cowan equations.

The promising results obtained on simple synthetic sounds, although preliminary, suggest possible applications of this framework to the problem of degraded speech. 
The next step will be to test the reconstruction ability of normal-hearing humans on originally degraded speech material compared to the same speech material after algorithm reconstruction. 
Such an endeavour will contribute to the understanding of the auditory mechanisms emerging in adverse listening conditions. It will furthermore
help to deepen our knowledge on general organization principles underlying the functioning of the human auditory cortex.

\section*{Acknowledgments}
The first three authors have been supported by the ANR project SRGI ANR-15-CE40-0018
and by the ANR project Quaco ANR-17-CE40-0007-01. 
This study was also supported by the IdEx Universite de Paris, ANR-18-IDEX-0001, awarded to the last author.


\appendix

\section{Integral kernel of the Kolmogorov operator}
\label{a:kernel}

The result in Proposition~\ref{p:expl-kernel} is well-known. E.g., by applying \cite[Proposition~9]{Barilari2017} and letting $x=(\freq,\acc)$ and $x'=(\freq',\acc')$ one gets that the kernel is 
\begin{equation}
	k_\delta(\freq,\nu\|\freq',\nu') 
	= \frac{1}{2\pi\sqrt{\det D_\delta}} \exp\left[-\frac{1}{2}
	\left(x'- \e^{\delta A}x\right)^{\top}D_\delta^{-1}	\left(x'-\e^{\delta A}	x\right)\right],
\end{equation}
%
%
where 
 \begin{equation}
 A=\begin{pmatrix}
 0&-1\\0&0
 \end{pmatrix} 
 \qquad\text{and}\qquad 
 B=\begin{pmatrix}
 0\\ \sqrt{2b}
 \end{pmatrix},
 \end{equation}
 and
\begin{equation}
D_\delta
=\e^{\delta A} \left[\int_{0}^{\tau } \e^{-\sigma A}BB^*\e^{-\sigma A^*}\diff \sigma \right]\e^{\tau A^*}.
\end{equation}

Direct computations yield
\begin{equation}
D_\delta
=2b 
	\begin{pmatrix}
	\delta^3/3 &-\tau^2/2
	\\
	-\delta^2/2&\tau
	\end{pmatrix}, 
	\quad\text{and}\quad \det D_\delta = \frac{b^2\tau^4}3.
\end{equation}
Therefore,
\begin{equation}
	\frac12 D_\delta^{-1} = \frac{1}{b \tau^3}M, 
	\quad\text{where}\quad
	M = \begin{pmatrix}
	3 & 3 \delta/2\\
	3 \delta/2&\delta^2
	\end{pmatrix}.
\end{equation}
Finally, the statement follows by letting
\begin{equation}
	g_\delta(x\|x') = \left(x'- \e^{\delta A}x\right)^{\top} M	\left(x'-\e^{\delta A}x\right).
\end{equation}

We now turn to an argument for Proposition~\ref{p:threshold}. Observe that $k_\delta(x\|x')\ge \varepsilon$ if and only if
\begin{equation}
	g_\delta(x\|x')\le \eta:=-b\delta^3\log\left( \frac{2\pi b\delta^2}{\sqrt 3} \varepsilon \right).
\end{equation}
	Then, we start by solving $z^\top M z\le \eta$, for $z\in \R^2$. One can check that this is verified if and only if
	\begin{equation}
		|z_2|\le \sqrt{\frac{4\eta}{\delta^2}},
		\quad\text{and}\quad
		\left| z_1 + \frac{\delta}{2} z_2 \right|\le \frac{\delta}{2\sqrt 3}\sqrt{\frac{4\eta}{\delta^2}-z_2^2}.
	\end{equation}
	Since $C_\varepsilon = 4\eta/\delta^2$ the statement follows by computing the above at $z = x'-e^{\tau A}x$.

\section{Heisenberg group action on the contact space}\label{a:heisenberg}

Recall that the short-time Fourier transform of a signal $s\in L^2(\R)$ is given by
\begin{equation}
  S(\tau, \omega) := \operatorname{STFT}(s)(\tau,\omega)= \int_{\R} s(t)W(\tau-t)e^{2\pi i t\omega}\,dt.
\end{equation}
Here, $W:\R\to [0,1]$ is a compactly supported (smooth) window, so that $S\in L^2(\R^2)$.
Fundamental operators in time frequency analysis \cite{Karlheinz2001} are time and phase shifts, acting on signals $s\in L^2(\R)$ by
\begin{equation}
  T_\theta s(t):= s(t-\theta)
  \quad\text{and}\quad
  M_\lambda s(t):=e^{2\pi i \lambda t}s(t),
\end{equation}
for $\theta,\lambda\in\R$.
One easily checks that $T_\theta$ and $M_\lambda$ are unitary operators on $L^2(\R)$. By conjugation with the short-time Fourier transform, they naturally define the unitary operators on $L^2(\R^2)$ given by
\begin{align}\label{eq:trans-stft}
  T_\theta S(\tau,\omega) &= e^{-2\pi i \omega \theta}S(\tau-\theta,\omega),\\
  M_\lambda S(\tau,\omega) &= S(\tau,\omega-\lambda).
\end{align}

The relevance of the Heisenberg group in time-frequency analysis is a direct consequence of the commutation relation
\begin{equation}
  T_\theta M_\lambda T_\theta^{-1} M_\lambda^{-1} = e^{-2\pi i\lambda\theta}\idty.
\end{equation}
Indeed, this shows that the operator algebra generated by $(T_\theta)_{\theta\in\R}$ and $(M_\lambda)_{\lambda\in\R}$ coincides with the Heisenberg group $\mathbb{H}^1$ via the representation $U:\mathbb{H}^1\to \mathcal{U}(L^2(\R^2))$ defined by
\begin{equation}\label{eq:heis-action}
  U(\theta,\lambda,\zeta)=e^{-2\pi i\zeta}T_\theta M_\lambda.
\end{equation}

The above discussion shows that the Heisenberg group can be regarded as the natural space of symmetries of sound signals. In particular, any meaningful treatment of these signals should respect such symmetry. In the case of our model, this is the content of the following result.

\begin{proposition}
    \label{prop:symmetry}
    The sound processing algorithm presented in this paper commutes with the Heisenberg group action \eqref{eq:heis-action} on sound signals. That is, if the input sound signal $s\in L^2(\R^2)$ yields $\hat{s}$ as a result, then, for any $(\theta,\lambda,\zeta)\in \mathbb{H}^1$, the input $U(\theta,\lambda,\zeta)s$ yields $U(\theta,\lambda,\zeta)\hat{s}$ as a result.
\end{proposition}

\begin{proof}
    We can schematically write the algorithm as:
    \begin{equation}
        \hat{s} = \operatorname{STFT}^{-1}\circ \operatorname{Proj} \circ \operatorname{WC}\circ \operatorname{Lift}\circ \operatorname{STFT}(s).
    \end{equation}
    Here, $\operatorname{Lift}$ is the lift operator defined in Section~\ref{sec:lift}, $\operatorname{WC}$ denotes the Wilson-Cowan evolution \eqref{eq:WC}, and $\operatorname{Proj}$ denotes the projection from the augmented space to the time-frequency representation, defined by
    \begin{equation}
        \operatorname{Proj}a (t,\omega) = \int_{-\infty}^{+\infty} a(t,\omega,\nu)\,d\nu. 
    \end{equation}

    Observe that \eqref{eq:trans-stft} shows that $U$ induces a representation of $\mathbb{H}^1$ on $L^2(\R^2)$, the codomain of the STFT, which we will denote by $\tilde U$. Thus, to prove the statement it suffices to show that
    \begin{equation}
        [\operatorname{Proj} \circ \operatorname{WC}\circ \operatorname{Lift}, \tilde U(\theta,\lambda,\zeta)] = 0, \qquad \forall (\theta,\lambda,\zeta)\in \mathbb{H}^1.
    \end{equation}

    Recall now that $\operatorname{Lift}$ associates with $S\in L^2(\R^2)$ a distribution of the form $\operatorname{Lift}[ S] (\tau,\freq,\acc)=S(\freq,\acc)\delta_\Sigma(\tau,\freq,\acc)$ for some $\Sigma\subset \R^3$. Due to the fact that $\Sigma$ is defined via the modulus of $S$, it is unaffected by the phase factors appearing in the representation $\tilde U$. That is, the lift of $\tilde U(\theta,\lambda,\zeta)S$ is given by
    \begin{equation}
  \begin{split}
      \operatorname{Lift}\big[\tilde U(\theta, \lambda,\zeta)S\big](\tau,\omega,\nu) 
        &= e^{-2\pi i (\zeta+\omega\theta)}\delta_{\Sigma}(\tau-\theta,\omega-\lambda,\acc)S(\tau-\theta,\freq-\lambda)\\
        &=e^{-2\pi i (\zeta+\omega\theta)}\operatorname{Lift}[S] (\tau-\theta,\omega-\lambda,\acc).
  \end{split}
\end{equation}
It is then immediate to check that $[\operatorname{Proj}\circ\operatorname{Lift}, \tilde U(\theta,\lambda,\zeta)]=0$.

We are left to verify that the operator $\operatorname{WC}$ commutes with $\tilde U(\theta,\lambda,\zeta)$. 
The commutation is trivial for the linear terms. On the other hand, the non-linearity introduced in the integral term commutes with $\tilde U(\theta,\lambda,\zeta)$ thanks to the fact that $\sigma(\rho e^{i\phi}) = e^{i\phi}\sigma(\rho)$ for all $\rho>0,\phi\in\R$.
\end{proof}


\bibliographystyle{abbrv}
\bibliography{biblio}

\end{document}